\documentclass[letterpaper, 10 pt, conference]{ieeeconf}  

\IEEEoverridecommandlockouts                              
\usepackage{amsmath, amssymb, amsfonts,bm}
\usepackage{mathrsfs}
\usepackage[mathcal]{euscript}
\usepackage{enumerate}
\usepackage{mathbbol}

\usepackage{float}

\usepackage{xcolor,graphicx,epstopdf}
\usepackage{dsfont}
\usepackage{hyperref}
\usepackage{caption}
\usepackage{subcaption}
\usepackage{cite} 

\usepackage{tikz-cd}
\usepackage{mathtools}

\usepackage[subfigure]{tocloft}
\usepackage{graphicx}
\usepackage{float}
\usepackage{epsfig}

\usepackage[ruled,vlined]{algorithm2e}

\usepackage{tikz}
\usepackage{pgfplots}
\usetikzlibrary{positioning,calc}

\usepackage{xcolor}
\definecolor{midnightblue}{rgb}{0, 0.14,0.55}
\usepackage[normalem]{ulem}

\newcommand*\mcapinn[2]{\vcenter{\hbox{$\mathsurround=0pt
			\ifx\displaystyle#1\textstyle\else#1\fi\bigcap$}}}

\newcommand*\mcupinn[2]{\vcenter{\hbox{$\mathsurround=0pt
			\ifx\displaystyle#1\textstyle\else#1\fi\bigcup$}}}

\newtheorem{theorem}{Theorem}
\newtheorem{definition}{Definition}

\newtheorem{lemma}{Lemma}

\newtheorem{proposition}{Proposition}
\newtheorem{assumption}{Assumption}
\newtheorem{example}{Example}

\title{\LARGE \bf
	Social Shaping of Dynamic Multi-Agent Systems over a Finite Horizon
}

\author{Zeinab Salehi, Yijun Chen, Ian R. Petersen, Elizabeth L. Ratnam, and Guodong Shi
	\thanks{This work was supported by the Australian Research Council under grants DP190102158, DP190103615, and LP210200473.}
	\thanks{Z. Salehi, I. R. Petersen and  E. L. Ratnam are with the Research School of Engineering, The Australian National University, Canberra, Australia. (E-mail:  zeinab.salehi@anu.edu.au; ian.petersen@anu.edu.au; elizabeth.ratnam@anu.edu.au) }
	\thanks{Y. Chen and G. Shi are with the  Australian Center for Field Robotics, The University of Sydney, NSW, Australia. (E-mail: yijun.chen@sydney.edu.au; guodong.shi@sydney.edu.au)}
}

\begin{document}

	\maketitle
	\thispagestyle{empty}
	\pagestyle{empty}

	\begin{abstract}
		This paper studies self-sustained dynamic multi-agent systems (MAS) for decentralized resource allocation operating at a competitive equilibrium over a finite horizon. The utility of resource consumption, along with the income from resource exchange, forms each agent’s payoff which is aimed to be maximized. Each utility function is parameterized by individual preferences which can be designed by agents independently. By shaping these preferences and proposing a set of utility functions, we can guarantee that the optimal resource price at the competitive equilibrium always remains socially acceptable, i.e., it never violates a given threshold that indicates affordability. First, we show this problem is solvable at the conceptual level under some convexity assumptions. Then, as a benchmark case, we consider quadratic MAS and formulate the associated social shaping problem as a multi-agent LQR problem which enables us to propose explicit utility sets using quadratic programming and dynamic programming. Finally, a numerical algorithm is presented for calculating the range of the preference function parameters which guarantee a socially accepted price. Some illustrative examples are given to examine the effectiveness of the proposed methods. 
		
	\end{abstract}

	\section{INTRODUCTION}
	Control and analysis of multi-agent systems (MAS) have received great attention among researchers due to their wide application in different areas such as economics \cite{zhang2020structured}, water systems \cite{negenborn2009distributed}, carbon markets \cite{narciso2021carbon}, robotics \cite{cheng2020safe}, power systems \cite{falsone2021new}, and smart grids \cite{pipattanasomporn2009multi}. 
	In general, a MAS consists of a group of agents who collaborate through a network to achieve a common objective or compete to reach an individual goal  \cite{deplano2018lyapunov}. 
	MAS are capable of solving  complex problems in a distributed manner  which are much harder or impossible to be solved by a single agent in a centralized way.

	
	One of the most fundamental problems in the literature is efficient resource allocation, which can be addressed by MAS approaches \cite{ebegbulem2018resource, lu2022distributed
	}. Depending on the application, there exist two common approaches: (i) social welfare where agents collaborate to maximize the total agents' utilities \cite{nguyen2014computational, chevaleyre}; (ii) competitive equilibrium in which agents compete to maximize their individual payoffs \cite{chen2021social, Bikhchandani1997}. In this paper, distributed resource allocation is treated as an optimization problem in which each agent maximizes its payoff under some constraints and the decision variable determines the amount of resource dedicated to each agent. Considering systems with dynamical states, this optimization problem becomes an optimal control problem. A competitive equilibrium, which is the pair of allocated resource and resource price, is proven to be an efficient solution to resource allocation problems by clearing the market \cite{wei2014competitive}.

	A fundamental theorem in classical welfare economics states that the competitive equilibrium is Pareto optimal, meaning that no agent can deviate from the equilibrium to achieve more profit without reducing other's payoff \cite{acemoglu2018, arrow1954, debrew1952}. It is also proved that under some convexity assumptions, the competitive equilibrium maximizes social welfare  \cite{mas-collel,Li2020,chen2021social}. Mechanism design is a well-known approach for social welfare maximization \cite{ma2020incentive}. For instance, the Groves mechanism  maximizes social welfare in a way that truth-telling of personal information by all agents forms a dominant strategy \cite{groves1973incentives}. The key point in achieving competitive equilibrium is efficient resource pricing that depends on the utility of each agent. The corresponding price, however, is not guaranteed to be affordable for all agents. If some participants select their utility functions aggressively, the price potentially increases to the point that it becomes unaffordable to other agents who have no alternative but to leave the system.  In such cases, the available resources are consumed by a limited number of affluent agents, which is not socially fair in societies where it is deemed that all entities are entitled to equal access \cite{salehi2021quadratic, salehi2021social}. A recent example is the Texas power outage disaster in February 2021, when some citizens who had access to electricity during the power outage received outstanding electricity bills for  their daily power usage, resulting in the dissatisfaction of customers \cite{Texas2021}.

	In this paper, we investigate how socially acceptable resource pricing at a competitive equilibrium is achievable  for self-sustained dynamic MAS  with distributed resource allocation over a finite horizon.  Agents allocate their resources in a way that their payoff, which consists of the utility from resource consumption and the income from resource exchange, is maximized over the whole horizon. The utility functions selected by agents affect resource pricing at the competitive equilibrium. By parameterizing these utility functions considering the preferences of agents and proposing some bounds on the parameters, we  control the resource price  such that it never exceeds a given threshold, so we achieve affordability. We  face an optimization problem and address it from  three points of view.
	\begin{itemize}
		\item A conceptual scheme, based on dynamic programming, is presented  to show  how the social shaping problem is solvable implicitly under some convexity assumptions for general classes of utility functions. Also, it is proved that when the price is positive, the total supply and demand are balanced across the network.
		
		\item The social shaping problem is reformulated for quadratic MAS, leading to an LQR problem. Solving the LQR problem using quadratic programming and dynamic programming, we propose two explicit sets for the preferences of agents which are proved to be socially admissible; i.e., they lead to socially acceptable resource prices.
		
		\item A numerical algorithm based on the bisection method is presented that provides accurate and practical bounds on the preferences of agents, followed by some convergence results.
	\end{itemize}
	
	The rest of the paper is organized as follows. In Section \ref{sec:preliminaries}, we review the multi-agent model, the system-level equilibria, and the concept of social shaping for dynamic MAS. In Section \ref{sec:conceptual}, we present a conceptual scheme for solving the  social shaping problem. In Section \ref{sec:Quadratic}, we introduce quadratic MAS and the LQR problem that follows. Then, we propose two explicit sets of agents' preferences using optimization methods. In Section \ref{sec:NumericalAlgorithm}, we present a numerical algorithm which provides accurate bounds on agents' preferences. Finally, Section \ref{sec:examples} includes some simulation results and Section \ref{sec:conclusions} contains conclusions.
	
	\textit{Notation:} We denote by $\mathbb{R}$ and $\mathbb{R}^{\geq 0}$ the fields of real numbers and  non-negative real numbers, respectively. $\mathbf I$
	is the identity  matrix with a suitable dimension. 
	The symbol $\mathbb 1$ represents a vector with an appropriate dimension whose entries are all $1$. We use $\left\| \mathbf \cdot \right\|$ to denote the Euclidean norm of a vector or its induced matrix norm.
	\section{Problem Formulation} \label{sec:preliminaries}
	In this section, we introduce the multi-agent system model, system level equilibria, and the concept of social shaping.
	
	
	\subsection{The Dynamic Multi-agent Model}\label{section:MAS}
	Consider a dynamic MAS with $n$  agents indexed in the set $\mathcal V=\{ 1, 2, \dots , n \}$. This MAS is studied in the time horizon $N$. Let time steps be indexed in the set $\mathcal{T}=\{ 0, 1, \dots , N-1 \}$. Each agent $i \in \mathcal{V}$ is a subsystem with dynamics represented by
	\begin{equation*}
		\mathbf x_i(t+1) = \mathbf A_i \mathbf x_i(t) + \mathbf B_i \mathbf u_i(t), \quad t \in \mathcal{T},
	\end{equation*}
	where $\mathbf x_i(t) \in \mathbb{R}^{d}$ is the dynamical state, $\mathbf x_i(0) \in \mathbb{R}^{d}$ is the given initial state, and $\mathbf u_i(t) \in \mathbb{R}^{m}$ is the control input. Also, $\mathbf A_i \in \mathbb{R}^{d \times d}$ and $\mathbf B_i \in \mathbb{R}^{d \times m}$ are fixed matrices. Upon reaching the state $\mathbf x_i(t)$ and employing the control input $\mathbf u_i(t)$ at time step $t \in \mathcal{T}$, each agent $i$ receives the utility $f_i(\mathbf x_i(t), \mathbf u_i(t)) = f(\cdot; \theta_i): \mathbb{R}^{d} \times \mathbb{R}^{m} \mapsto \mathbb{R}$, where $\theta_i \in \Theta$ is a personalized parameter of agent $i$. The terminal utility achieved as a result of reaching the terminal state $\mathbf x_i(N)$ is denoted by  $\phi_i(\mathbf x_i(N)) = \phi(\cdot; \theta_i): \mathbb{R}^{d} \mapsto \mathbb{R}$.
	At each time step $t \in \mathcal{T}$,  agent $i$ provides a local energy supply $a_i(t) \in \mathbb{R}^{\geq 0}$,  and consumes an amount of energy $h_i(\mathbf u_i(t)): \mathbb{R}^{m} \mapsto \mathbb{R}^{\geq 0}$  as a result of taking the control action $\mathbf u_i(t)$. The overall network supply $C(t)>0$ is then defined as $C(t) := \sum_{i=1}^n a_i(t)$ for $t \in \mathcal{T}$. Agents are interconnected through a network to sell (or buy) their surplus (or shortage) of energy. This means each agent $i \in \mathcal{V}$ can decide how much of their extra energy $a_i(t)- h_i(\mathbf u_i(t))$ is going to be traded through the network leading to a new decision variable named \emph{strategic trading decision}, denoted by $e_i(t) \in \mathbb{R}$. There is a physical constraint indicating the traded resource for each agent can never be greater than the surplus of resource, i.e., $e_i(t) \leq a_i(t)- h_i(\mathbf u_i(t))$.
	The price for unit resource exchange across the network at each time step $t \in \mathcal{T}$ is denoted by $\lambda_t$. Then, the income or cost from resource exchange for agent $i$ is represented by  $\lambda_t e_i(t)$. 
	\subsection{System-level Equilibria}
	Let $\mathbf U_i=(\mathbf u_i^\top(0), \dots, \mathbf u_i^\top(N-1))^\top$ and $\mathbf E_i=( e_i(0), \dots,  e_i(N-1))^\top$ denote the vector of control inputs and the vector of strategic trading decisions associated with agent $i$ over the whole time horizon, respectively. Also, let $\mathbf u(t)= (\mathbf u_1^\top(t), \dots, \mathbf u_n^\top(t) )^\top$ and $\mathbf e(t)= ( e_1(t), \dots,  e_n(t) )^\top$ denote the vector of  control inputs and the vector of  strategic trading decisions associated with all agents at time step $t \in \mathcal{T}$, respectively. Let $\mathbf U=(\mathbf u^\top(0), \dots, \mathbf u^\top(N-1))^\top$ and $\mathbf E=(\mathbf e^\top(0), \dots, \mathbf e^\top(N-1))^\top$ be the vector of all control inputs and the vector of all strategic trading decisions at all time steps, respectively. Let $\boldsymbol  \lambda = (\lambda_0, \dots, \lambda_{N-1})^\top$ denote the vector of resource prices throughout the entire time horizon.
	
	
	\begin{definition}
		The \textit{competitive equilibrium } for a dynamic MAS is the triplet $(\boldsymbol{\lambda}^\ast, \mathbf{U}^\ast, \mathbf{E}^\ast)$ which satisfies the following two conditions.
		\begin{itemize}
			\item[(i)]  Given $\boldsymbol \lambda^\ast$, the pair $(\mathbf U^\ast, \mathbf E^\ast)$ maximizes the individual payoff function of each agent; i.e., each $(\mathbf U_i^\ast, \mathbf E_i^\ast)$ solves the following constrained maximization problem  
			\begin{equation}\label{opt_DLTD_1}
			\resizebox {0.447\textwidth} {!} {$
				\begin{aligned}
					\max_{{\mathbf U}_i, {\mathbf E}_i} \, \, \, &  \phi(\mathbf x_i(N); \theta_i) + \sum_{t=0}^{N-1} \Big( f( \mathbf x_i(t), \mathbf u_i(t); \theta_i)+\lambda_t^\ast e_i(t) \Big) \\
					{\rm s.t.} \quad &\mathbf x_i(t+1) = \mathbf A_i \mathbf x_i(t)+ \mathbf B_i \mathbf u_i(t),\\
					& e_i(t)  \leq a_i(t) - h_i(\mathbf u_i(t)), \quad t \in \mathcal{T}. 
				\end{aligned}$}
			\end{equation}
			\item[(ii)] The optimal strategic trading $\mathbf E^\ast$ balances the total traded resource across the network at each time step; that is,
			\begin{equation}\label{trading_demand_supply_constraints}
				\sum_{i=1}^n e_i^\ast(t) =0, \quad t \in \mathcal{T}. 
			\end{equation} 
		\end{itemize}
	\end{definition}
	
	
	\medskip
	
	\begin{definition}
		The \textit{social welfare equilibrium} for a dynamic MAS is the pair $ (\mathbf{U}^\ast, \mathbf{E}^\ast)$ which solves the following optimization problem
		\begin{equation}\label{opt_social_DLD_1}
			\begin{aligned}
				\max_{\mathbf{U}, \mathbf{E}} \quad &  \sum_{i=1}^{n} \left( \phi(\mathbf x_i(N); \theta_i) + \sum_{t=0}^{N-1} f(\mathbf x_i(t), \mathbf u_i(t); \theta_i) \right) \\
				{\rm s.t.} \quad & \mathbf x_i(t+1) = \mathbf A_i \mathbf x_i(t)+ \mathbf B_i \mathbf u_i(t), \\
				& e_i(t)  \leq a_i(t) - h_i(\mathbf u_i(t)), \\
				&\sum_{i=1}^n e_i(t) =0, 
				\quad   t \in \mathcal{T}, \,\,\, i \in \mathcal V. 
			\end{aligned}
		\end{equation}
		In the social welfare equilibrium, the total agent utility functions are maximized.
	\end{definition}
	
	\medskip
	
	\begin{assumption}\label{assumption1}
		$f(\cdot; \theta_i)$ and $\phi(\cdot; \theta_i)$ are concave functions for all $i \in \mathcal{V}$. Additionally,  $h_i(\cdot)$ is a non-negative convex function such that $h_i(\mathbf z)<b$ represents a bounded open set of $\mathbf z$ in $\mathbb{R}^m$ for $b>0$ and $i \in \mathcal{V}$. Furthermore, assume $\sum_{i=1}^{n} a_i(t)>0$ for all $t \in \mathcal{T}$.
	\end{assumption}
	
	\medskip
	
	\begin{proposition}[as in \cite{chen2021social}]\label{prop1}
		Suppose Assumption \ref{assumption1} holds. Then the competitive equilibrium and the social welfare equilibrium exist and coincide. Additionally, the optimal price $\lambda^\ast_t$ in \eqref{opt_DLTD_1} is obtained from the Lagrange multiplier associated with the balancing equality constraint $\sum_{i=1}^n e_i(t) =0$ in \eqref{opt_social_DLD_1}. 
	\end{proposition}
	
	\subsection{Social Shaping Problem}
	The optimal price $\lambda^\ast_t$, which is the Lagrange multiplier corresponding to the equality constraint in \eqref{opt_social_DLD_1}, depends on the utility functions of agents. If there are no regulations on the choice of utility functions, the price may become extremely high and unaffordable for some agents. In this case, those who have found the  price unaffordable cannot compete in the market and have no alternative but to leave the system. Consequently, all of the resources will be consumed by a limited number of agents who have dominated the price by aggressively selecting their utilities. This is indeed  socially unfair and not sustainable. 
	So we need a mechanism, called  social shaping,  which ensures the price is always below an acceptable threshold denoted by $\lambda^\dag \in \mathbb{R}^{>0}$. The problem of social shaping is addressed for static MAS in \cite{salehi2021quadratic} and \cite{salehi2021social}. Now, we define an extended version of the  social shaping problem for dynamic MAS as follows.

	\begin{definition}[Social shaping for dynamic MAS]
		Consider a dynamic MAS whose agents $i \in \mathcal{V}$ have $f(\cdot; \theta_i)$ and $\phi(\cdot; \theta_i)$ as their running utility function and terminal utility function, respectively. Let $\lambda^\dag \in \mathbb{R}^{>0}$ be the given price threshold accepted by all agents. Find a range $\Theta$ of personal parameters $\theta_i$ such that if  $\theta_i \in \Theta$ for $i \in \mathcal{V}$, or $(\theta_1, \dots, \theta_n) \in \Theta^n$,
		 then we yield $\lambda_t^\ast \leq \lambda^\dag$ at all time steps $t \in \mathcal{T}$.
	\end{definition}
	\section{Conceptual Social Shaping}\label{sec:conceptual}
	In this section, we examine how the social shaping problem of dynamic MAS can be solved conceptually.
	
	\medskip
	
	\begin{lemma}\label{lemma1}
		Consider the dynamic MAS. If Assumption \ref{assumption1} is satisfied, then $\lambda^\ast_t \geq 0$ for all $t \in \mathcal{T}$.
	\end{lemma}
	\begin{proof}
		The proof is similar to the proof of Proposition 2 in \cite{chen2021social}. 
	\end{proof}
	
	\medskip 
	
	\begin{proposition}\label{prop2}
		Consider the dynamic MAS. Let Assumption \ref{assumption1} hold. If $\lambda^\ast_t > 0$ then the total demand and supply are balanced at time step $t$; that is,
		\begin{equation}\label{eq3}
			\sum_{i=1}^{n} h_i(\mathbf u_i^\ast(t)) = C(t).
		\end{equation}
	\end{proposition}
	\begin{proof}
		Since Assumption \ref{assumption1} is satisfied, Proposition \ref{prop1} holds. Therefore, either the competitive problem or the social welfare problem can be solved. Consider the competitive optimization problem in \eqref{opt_DLTD_1}.  Let $s_i(t) \in \mathbb{R}^{\geq 0}$ be the slack variable for agent $i$ at time step $t \in \mathcal{T}$ and $\mathbf{S_i} = (s_i(0), ..., s_i(N-1))^\top$ be the vector of slack variables for agent $i$ throughout the whole time horizon.
		We can write the inequality constraint $e_i(t) \leq a_i(t)- h_i(\mathbf u_i(t))$ in \eqref{opt_DLTD_1} as the equality $e_i(t) = a_i(t) - h_i(\mathbf u_i(t)) -  s_i(t)$ for  $t \in \mathcal{T}$.
		Then, substituting $e_i(t) = a_i(t) - h_i(\mathbf u_i(t)) - s_i(t)$ into \eqref{opt_DLTD_1} results in an equivalent form for the optimization problem as
		\begin{equation}\label{opt_DLTD_1_2}
			\small
			\begin{aligned}
				\max_{{\mathbf U}_i, {\mathbf S}_i} \quad &  \phi(\mathbf x_i(N); \theta_i) + \sum_{t=0}^{N-1}  f( \mathbf x_i(t), \mathbf u_i(t); \theta_i) \\ & + \sum_{t=0}^{N-1} \lambda_t^\ast \big[ a_i(t) - h_i(\mathbf u_i(t)) - s_i(t) \big]\\
				{\rm s.t.} \quad & \mathbf x_i(t+1) = \mathbf A_i \mathbf x_i(t)+ \mathbf B_i \mathbf u_i(t),\\
				&s_i(t) \in \mathbb{R}^{\geq0}, \quad t \in \mathcal{T}.
			\end{aligned}
		\end{equation}
		Since $\lambda^\ast_t >0$, the resulting objective function is strictly decreasing with respect to $s_i(t)$. Consequently, the optimal slack variable maximizing the objective function is $s_i^\ast(t) =0$, meaning that the associated inequality constraint is active; that is,
		\begin{equation}\label{eq2}
			e_i^\ast(t) = a_i(t) - h_i(\mathbf u^\ast_i(t)).
		\end{equation}
		The summation of \eqref{eq2} over $i$, from $1$ to $n$, along with the balancing equality $\sum_{i=1}^{n} e_i^\ast(t)=0$ in \eqref{trading_demand_supply_constraints}, yields $\sum_{i=1}^{n}h_i(\mathbf u^\ast_i(t)) = C(t)$.
	\end{proof}
	
	\medskip
	
	Now, let us work out how the social shaping problem can be solved conceptually. Suppose Assumption \ref{assumption1} holds and $f(\cdot; \theta_i)$, $\phi(\cdot; \theta_i)$, and $h_i(\cdot)$ are continuously differentiable. Then Proposition \ref{prop1} is satisfied. In this paper, we focus on  the competitive optimization problem in \eqref{opt_DLTD_1}. 
	According to Lemma \ref{lemma1}, there holds $\lambda^\ast_t \geq 0$. We can skip the case $\lambda_t^\ast=0$, because a zero price is always socially resilient. Therefore, it is sufficient to only examine $\lambda^\ast_t >0$. Following from Proposition \ref{prop2}, the total demand and supply are balanced at time step $t$, meaning that $\sum_{i=1}^{n} h_i(\mathbf u_i^\ast(t))= C(t)$; additionally, we have $s_i^\ast(t)=0$. Substituting  $s_i(t) =0$ into \eqref{opt_DLTD_1_2} yields an equivalent form for the competitive optimization problem in \eqref{opt_DLTD_1} as
	\begin{equation}\label{opt_DLTD_1_3}
		\begin{aligned}
			\max_{{\mathbf U}_i} \quad &  \phi(\mathbf x_i(N); \theta_i ) + \sum_{t=0}^{N-1}  f( \mathbf x_i(t), \mathbf u_i(t); \theta_i )\\
			&+ \sum_{t=0}^{N-1} \lambda_t^\ast \big[ a_i(t) - h_i(\mathbf u_i(t)) \big]\\
			{\rm s.t.} \quad & \mathbf x_i(t+1) = \mathbf A_i \mathbf x_i(t)+ \mathbf B_i \mathbf u_i(t), \quad t \in \mathcal{T},
		\end{aligned}
	\end{equation}
	which is valid for $\lambda^\ast_t>0$. Please note that even if $\lambda^\ast_t = 0$ then \eqref{opt_DLTD_1_2} can be written in the form of \eqref{opt_DLTD_1_3}, although the equality in \eqref{eq3} turns into the inequality $\sum_{i=1}^{n}h_i(\mathbf u^\ast_i(t)) \leq C(t)$. This fact causes no change to the upcoming analysis. 
	The optimization problem in \eqref{opt_DLTD_1_3} is in essence an unconstrained optimal control problem 
	which can be solved by the dynamic programming approach. First, introduce the cost-to-go function for agent $i$ from time $k$ to $N$ as
	\begin{equation*}
		\resizebox {0.485\textwidth} {!} {$
			\begin{aligned}
				&J_{i}^{k \longrightarrow N} (\mathbf x_i(k), \mathbf u_i(k),  \dots, \mathbf u_i(N-1), \lambda^\ast_k, \dots, \lambda^\ast_{N-1}; \theta_i ) \\ &= \phi ( \mathbf x_i(N); \theta_i ) + \sum_{t=k}^{N-1} \Big( f (\mathbf x_i(t), \mathbf u_i(t); \theta_i )+\lambda_t^\ast \big[ a_i(t) - h_i (\mathbf u_i(t)) \big]\Big).
			\end{aligned} $}
	\end{equation*}
	Then, the optimal cost-to-go at time $k$ for agent $i$, which is also called the value function, is represented as
	\begin{equation*}
		\resizebox {0.485\textwidth} {!} {$
			\begin{aligned}
				&V_{i, k}(\mathbf x_i(k), \lambda^\ast_k, \dots, \lambda^\ast_{N-1}; \theta_i) \\ &= \max_{\mathbf u_i(k), \dots, \mathbf u_i(N-1)} \quad   J_{i}^{k \longrightarrow N} (\mathbf x_i(k), \mathbf u_i(k), \dots, \mathbf u_i(N-1), \lambda^\ast_k, \dots, \lambda^\ast_{N-1}; \theta_i) \\
				& \quad \quad \quad \quad  {\rm s.t.} \, \, \,  \quad \quad \quad  \mathbf x_i(t+1) = \mathbf A_i \mathbf x_i(t)+ \mathbf B_i \mathbf u_i(t), \quad t= k, \dots, N-1.
			\end{aligned} $}
	\end{equation*}
	Now, we go one step backward in time. According to the principle of optimality, the optimal cost-to-go from time  $k-1$ to $N$ is obtained by optimizing the sum of the running cost at time $k-1$ and the value function at time $k$; that is,
	\begin{equation}\label{eq12}
		\resizebox {0.487\textwidth} {!} {$
			\begin{aligned}
				&V_{i, k-1}(\mathbf x_i(k-1), \lambda^\ast_{k-1}, \dots, \lambda^\ast_{N-1}; \theta_i) \\ &= \max_{\mathbf u_i(k-1)} \quad   f (\mathbf x_i(k-1), \mathbf u_i(k-1); \theta_i)+\lambda_{k-1}^\ast \big[ a_i(k-1) - h_i \big(\mathbf u_i(k-1)\big) \big] \\& \quad \quad \quad \quad \quad + V_{i, k}(\mathbf x_i(k), \lambda^\ast_{k}, \dots, \lambda^\ast_{N-1}; \theta_i)\\
				& \quad \quad \, \, \, {\rm s.t.} \quad \, \, \, \mathbf x_i(k) = \mathbf A_i \mathbf x_i(k-1)+ \mathbf B_i \mathbf u_i(k-1).
			\end{aligned} $}
	\end{equation}
	Based on the system dynamics, $\mathbf x_i(k)$ depends on $\mathbf x_i(k-1)$ and $\mathbf u_i(k-1)$. Consequently, the only decision variable in \eqref{eq12} is $\mathbf u_i(k-1)$. The resulting value function depends on $\mathbf x_i(k-1)$ and $\lambda^\ast_t$ where $t=k-1, \dots, N-1$. 
	Now, let us start with the final time $k=N$. The terminal value function is
	\begin{equation*}\label{eq14}
		\begin{aligned}
			V_{i, N}(\mathbf x_i(N); \theta_i) = \phi \big( \mathbf x_i(N); \theta_i \big).
		\end{aligned}
	\end{equation*}
	Moving backward in time and substituting $\mathbf x_i(k) = \mathbf A_i \mathbf x_i(k-1)+ \mathbf B_i \mathbf u_i(k-1)$ into the cost function, we obtain
	\begin{equation*}\label{eq17}
		\begin{aligned}
			&V_{i, N-1}(\mathbf x_i(N-1), \lambda^\ast_{N-1}; \theta_i) \\&= \max_{\mathbf u_i(N-1)} \quad   f (\mathbf x_i(N-1), \mathbf u_i(N-1); \theta_i) \\& \quad \quad \quad \quad \quad +\lambda_{N-1}^\ast \big[ a_i(N-1) - h_i \big(\mathbf u_i(N-1)\big) \big]  \\& \quad \quad \quad \quad \quad + V_{i, N}(\mathbf A_i \mathbf x_i(N-1)+ \mathbf B_i \mathbf u_i(N-1); \theta_i),\\
			\vdots\\
			&V_{i, 0}(\mathbf x_i(0), \lambda_{0}^\ast, \dots, \lambda_{N-1}^\ast; \theta_i) \\ &= \max_{\mathbf u_i(0)} \quad   f (\mathbf x_i(0), \mathbf u_i(0); \theta_i)+ \lambda_{0}^\ast \big[ a_i(0) - h_i \big(\mathbf u_i(0)\big) \big] \\ & \quad \quad \quad \quad \quad + V_{i, 1}(\mathbf A_i \mathbf x_i(0)+ \mathbf B_i \mathbf u_i(0), \lambda_{1}^\ast, \dots, \lambda_{N-1}^\ast; \theta_i).
		\end{aligned}
	\end{equation*}
	To obtain the optimal control at time step $k=0$, the derivative of the associated objective function with respect to $\mathbf u_i(0)$ must equal zero; that is,
	\begin{multline*}
		\frac{\partial f(\mathbf x_i(0), \mathbf u_i(0);\theta_i)}{\partial \mathbf u_i(0)}- \lambda_{0}^\ast \nabla h_i \big(\mathbf u_i(0)\big)\\ + \frac{\partial V_{i,1}(\mathbf A_i\mathbf x_i(0) +\mathbf B_i \mathbf u_i(0), \lambda_{1}^\ast, \dots, \lambda_{N-1}^\ast;\theta_i)}{\partial \mathbf u_i(0)}=0.
	\end{multline*} 
	Proposition \ref{prop1} implies that such an optimal solution exists, although it might not be unique. {Without loss of generality, suppose the optimal solution is unique.} Thus, we can write $\mathbf u_i^\ast(0)$ as a function of $\mathbf x_i(0)$ and all $\lambda^\ast_t$ where $t \in \mathcal{T}$,   parameterized by $\theta_i$; that is,
	\begin{equation}\label{eq20}
		\mathbf u_i^\ast(0)=l_i^0(\mathbf x_i(0), \lambda_{0}^\ast, \dots, \lambda_{N-1}^\ast; \theta_i).
	\end{equation}
	Substituting \eqref{eq20} into the  equality $\sum_{i=1}^{n} h_i(\mathbf u_i^\ast(t)) = C(t)$ in \eqref{eq3}, we yield
	\begin{equation*}\label{eq18}
		\sum_{i=1}^{n} h_i(l_i^0(\mathbf x_i(0), \lambda_{0}^\ast, \dots, \lambda_{N-1}^\ast; \theta_i)) = C(0).
	\end{equation*}
	Similarly, for any other time step $k \in \mathcal{T}$ we  achieve
	\begin{equation*}\label{eq33}
		\mathbf u_i^\ast(k)= l_i^k(\mathbf x_i(0), \lambda^\ast_0, \dots, \lambda_{N-1}^\ast; \theta_i), \,\,\,\,  k \in \mathcal{T}, 
	\end{equation*}
	and
	\begin{equation}\label{eq32}
			\sum_{i=1}^{n} h_i( l_i^k(\mathbf x_i(0), \lambda^\ast_0, \dots, \lambda_{N-1}^\ast; \theta_i)) = C(k), \,\,\,\,  k \in \mathcal{T}. 
	\end{equation}
	We aim to obtain $\boldsymbol{\lambda^\ast} = (\lambda^\ast_0, \dots, \lambda^\ast_{N-1})^\top$. According to \eqref{eq32}, we have $N$ equations with $N$ variables. 
	Let $\mathbf x(k) = (\mathbf x^\top_1(k), \dots, \mathbf x^\top_n(k))^\top$, $\boldsymbol{\theta}=(\theta_1, \theta_2, \dots, \theta_n)$, and $\mathbf C =(C(0), C(1), \dots, C(N-1))$. According to Proposition \ref{prop1}, there exists $\boldsymbol \lambda^\ast$ which satisfies \eqref{eq32} although it might not be unique. Among different possible prices that satisfy the equilibrium, we consider the maximum one at each time step. In the rest of this paper, by optimal price we mean the maximum possible price associated with a fixed $\boldsymbol{\theta}$ meeting the equilibrium conditions. Solving \eqref{eq32}, the optimal price at each time step $k \in \mathcal{T}$ is obtained as
	\begin{equation*}\label{eq19}
		\lambda^\ast_k= g_k(\mathbf x(0), \mathbf C; \boldsymbol \theta), \quad k=0, \dots, N-1.
	\end{equation*}
	Additionally, for different values of agent preferences $\boldsymbol \theta$ we would obtain different optimal prices at each time step. Let us define the maximum value of the set of all possible optimal prices at each time step $k$, when $\theta_i$ takes values in the set $\Theta$ (or $\boldsymbol{\theta} \in \Theta^n$), as
	\begin{equation*}
		\chi^\Theta_k : = \max_{\boldsymbol \theta \in \Theta^n} g_k(\cdot; \boldsymbol \theta), \quad k=0, \dots, N-1.
	\end{equation*} 
	Next, we introduce
	\begin{equation*}
		\mathbf G_\Theta:=\big(	\chi^\Theta_0, \chi^\Theta_1, \dots, \chi^\Theta_{N-1} \big)^\top.
	\end{equation*}
	Each element $k$ in the vector $\mathbf G_\Theta$ is the maximum value of optimal prices at  time step $k$, when agent preferences are taken from $\boldsymbol \theta \in \Theta^n$.  This leads to the following result.

	\medskip
	
	\begin{theorem}
		Consider a dynamic MAS. Let Assumption \ref{assumption1} hold. 
		Suppose $f(\cdot; \theta_i)$, $\phi(\cdot; \theta_i)$, and $h_i(\cdot)$ are continuously differentiable. Let $\lambda^\dag \in \mathbb{R}^{>0}$ represent the given price threshold accepted by all agents.  
		Then any set $\Theta$ satisfying $\mathbf G_\Theta \leq \lambda^\dag \mathbb{1}$ ensures that $\lambda^\ast_t \leq \lambda^\dag$ for $t \in \mathcal{T}$, and thus, solves the social shaping problem of agent preferences. 
	\end{theorem}
	
	\section{Quadratic Social Shaping}\label{sec:Quadratic}
	In this section, we  examine quadratic utility functions for dynamic MAS and explicitly propose two sets of personal parameters which guarantee that the optimal prices at all time steps are socially resilient. 
	
	\medskip
	
	\begin{assumption}\label{assumption2}
		Consider the dynamic MAS introduced in Section \ref{section:MAS}. Let $\theta_i:= (\mathbf Q_i, \mathbf R_i)$, where $\mathbf Q_i \in \mathbb{R}^{d \times d}$, $\mathbf Q_i=\mathbf Q_i^\top>0$ and $\mathbf R_i \in \mathbb{R}^{m \times m}$, $\mathbf R_i=\mathbf R_i^\top>0$.
		Assume for all $i\in \mathcal V$ we have	
		\begin{equation*}\label{eq_utility_q}
			\begin{gathered}
				\begin{aligned}	
					f(\mathbf x_i(t), \mathbf u_i(t); \theta_i) &= - \mathbf x_i^\top(t) \mathbf Q_i \mathbf x_i(t)  - \mathbf u_i^\top(t) \mathbf R_i \mathbf u_i(t),\\
					\phi(\mathbf x_i(N); \theta_i) &= - \mathbf x_i^\top(N) \mathbf Q_i \mathbf x_i(N),\\
					h_i(\mathbf u_i(t)) &=  \mathbf u_i^\top(t)\mathbf H_i \mathbf u_i(t),
				\end{aligned}
			\end{gathered}
		\end{equation*}
		where $\mathbf H_i \in \mathbb{R}^{m \times m}$, $\mathbf H_i=\mathbf H_i^\top>0$. 
	\end{assumption}

	\medskip 
	
	\begin{assumption}\label{assumption3}
		Consider the dynamic MAS in Assumption \ref{assumption2} with a given initial state $\mathbf x_i(0)$ such that  $\left\| {{\mathbf x_i}(0)} \right\| \le \gamma$, $\left\|\mathbf A_i \right\| \leq \alpha$, $\left\| \mathbf B_i \right\| \leq \beta$, and $\mathbf H_i \geq \rho \mathbf I$ for $i \in \mathcal V$. Suppose that $\gamma, \alpha, \beta, \rho \in \mathbb{R}^{>0}$.
	\end{assumption}

	\medskip 
	
	We aim to solve the following social shaping problem.
	
	\medskip
	
	\noindent{\bf Dynamic \& Quadratic Social Shaping Problem.} Suppose Assumptions \ref{assumption2} and \ref{assumption3} hold. Let $\lambda^\dag \in \mathbb{R}^{>0}$ be the given price threshold accepted by all agents, and $\delta_{\rm max}  \in \mathbb{R}^{>0}$ be an upper bound for the norm of the personal parameter $\mathbf Q_i$. We propose an admissible set for $\delta_{\rm max}$ such that all utility functions satisfying $\left\| \mathbf Q_i \right\| \leq \delta_{\rm max}$  (or $\mathbf Q_i \leq \delta_{\rm max} \mathbf I$) lead to socially acceptable energy prices at all time steps; i.e., $ \lambda^\ast_t \leq \lambda^\dag$ for $t \in \mathcal{T}$. 
	
	\medskip
	
	To address this problem, we use two approaches:  quadratic programming and dynamic programming.
	
	\subsection{Quadratic Programming Approach}\label{sec:QuadProg}
	Since Assumption \ref{assumption2} is satisfied, Proposition \ref{prop1} holds. We examine the competitive optimization problem in \eqref{opt_DLTD_1}. According to Lemma \ref{lemma1}, there holds $\lambda^\ast_t \geq 0$. We skip the case $\lambda_t^\ast=0$, because a zero price is always socially resilient. Hence, it is sufficient to only study $\lambda^\ast_t >0$.  According to \eqref{opt_DLTD_1_3}, the optimization problem in \eqref{opt_DLTD_1} can be reformulated as 
	\begin{equation}\label{opt_DLTD_2}
		\begin{aligned}
			\max_{{\mathbf U}_i} \quad &- \mathbf x_i^\top(N) \mathbf Q_i \mathbf x_i(N) + \sum_{t=0}^{N-1} \big[ \big. - \mathbf x_i^\top(t) \mathbf Q_i \mathbf x_i(t)   \\  &-\mathbf u_i^\top(t) \mathbf R_i \mathbf u_i(t)  +  \lambda_t^\ast \big(a_i(t) -  \mathbf u_i^\top(t)\mathbf H_i \mathbf u_i(t) \big) \big. \big]  \\
			{\rm s.t.} \quad & \mathbf x_i(t+1) = \mathbf A_i \mathbf x_i(t)+ \mathbf B_i \mathbf u_i(t), \quad   t \in \mathcal{T}. \\ 
		\end{aligned}
	\end{equation}

	\begin{theorem}\label{theorem2}
		Consider the dynamic MAS described in Assumptions \ref{assumption2} and \ref{assumption3} on the time horizon $N$.
		Suppose $\delta_{\rm max} \in \mathbb{R}^{>0}$ is selected from the following set
		\begin{equation*}\label{eq_set_Q1}
					\begin{aligned}
						&\mathscr{S}_{\ast} =\left\{ \Bigg. \right. \delta_{\rm max} \in \mathbb{R}^{>0} : 
						\delta_{\rm max} \sum_{t=k+1}^{N}  \Bigg[ \Bigg. \gamma \alpha^{2t-k-1}   \\ &+ \beta   { \sum_{\substack{ j=0 \\ j\ne k}}^{t-1}} \sqrt{\frac{C(j)}{\rho}} \alpha^{2t-j-k-2}  \Bigg.\Bigg]  \leq \frac{\sqrt{C(k)\rho}}{n \beta}\lambda^\dag \, \, \, \, {\rm for}\,\, \forall k \in \mathcal{T} \left. \Bigg. \right\}.
					\end{aligned}
		\end{equation*} 
		Then for all quadratic utility functions satisfying $\left\| \mathbf Q_i \right\| \leq \delta_{\rm max}$  (or $\mathbf Q_i \leq \delta_{\rm max} \mathbf I$), the resulting optimal price is socially resilient, i.e., $\boldsymbol \lambda^\ast \leq \lambda^\dag \mathbb{1}$.
	\end{theorem}	
	\begin{proof}
		Considering the  equality in \eqref{eq3}, we obtain
		\begin{equation}\label{eq3_2}
			\sum_{i=1}^{n} \mathbf u_i^{\ast \top}(t) \mathbf H_i \mathbf u_i^\ast(t) = C(t), \quad t \in \mathcal{T}.
		\end{equation}
	Furthermore, the following inequality holds
	\begin{equation}\label{eq3_3}
		\sigma_{\rm min}(\mathbf H_i) \left\| \mathbf u_i^\ast(t) \right\|^2 \leq \mathbf u_i^{\ast \top}(t) \mathbf H_i \mathbf u_i^\ast(t),
	\end{equation}
	where $\sigma_{\rm min}(\mathbf H_i)$ is the smallest eigenvalue of $\mathbf H_i$.
		Additionally, since $\mathbf u_i^{\ast \top}(t) \mathbf H_i \mathbf u_i^\ast(t) \geq 0$, the equality in \eqref{eq3_2} yields 
		$	\mathbf u_i^{\ast \top}(t) \mathbf H_i \mathbf u_i^\ast(t) \leq C(t).$
		Following from \eqref{eq3_3}, we obtain
		\begin{equation}\label{eq3_5}
			\sigma_{\rm min}(\mathbf H_i) \left\| \mathbf u_i^\ast(t) \right\|^2 \leq C(t).
		\end{equation}
		According to Assumption \ref{assumption3}, we have $\mathbf H_i \geq \rho I$, meaning that $\sigma_{\rm min}(\mathbf H_i) \geq \rho$. Consequently, inequality \eqref{eq3_5} results in 
		\begin{equation}\label{eq3_6}
			\left\|\mathbf u_i^\ast(t) \right\| \leq \sqrt{\frac{C(t)}{\rho}}.
		\end{equation}
		Additionally, from the dynamical equation in \eqref{opt_DLTD_2}, we obtain
		\begin{equation}\label{eq15}
			\mathbf x_i(t)=\mathbf A_i^t \mathbf x_i(0)+ \sum_{j=0}^{t-1}{\mathbf A_i^{t-j-1} \mathbf B_i \mathbf u_i(j)}, \quad t \in \{1, 2, ..., N\}.
		\end{equation}
		Substituting \eqref{eq15} into \eqref{opt_DLTD_2} yields an unconstrained optimization problem, in which the only decision variable is ${\mathbf U}_i$. The associated objective function $J$ is 
		\begin{equation*}\label{opt_DLTD_4}
				\small
				\begin{split}
					&J :=  \sum_{t=1}^{N} \Bigg[ \Bigg. - \Bigg( \mathbf A_i^t \mathbf x_i(0) + \sum_{j=0}^{t-1} \mathbf A_i^{t-j-1} \mathbf B_i \mathbf u_i(j) \Bigg)^\top \mathbf Q_i \\ &\times \Bigg( \mathbf A_i^t \mathbf x_i(0) + \sum_{j=0}^{t-1} \mathbf A_i^{t-j-1} \mathbf B_i \mathbf u_i(j) \Bigg)  \Bigg.\Bigg]  - \mathbf x_i^\top(0) \mathbf Q_i \mathbf x_i(0)   \\   &+\sum_{t=0}^{N-1} \Bigg[ \Bigg. - \mathbf u_i^\top(t) \mathbf R_i \mathbf u_i(t) +\lambda_t^\ast \Bigg(  a_i(t) -  \mathbf u_i^\top(t) \mathbf H_i \mathbf u_i(t) \Bigg) \Bigg.\Bigg]. 
				\end{split} 
		\end{equation*}
		Let $k$ denote a time interval indexed in $\mathcal{T}$. 
		Let $\frac{\partial J}{\partial \mathbf u_i(k)} =0$, which yields 
		\begin{multline}\label{eq21}
		\lambda^\ast_k \mathbf H_i \mathbf u_i(k)=-\sum_{t=k+1}^{N}\Bigg[ \Bigg. \Bigg(\mathbf A_i^{t-k-1} \mathbf B_i \Bigg)^\top  \mathbf Q_i \\ \times \Bigg( \mathbf A_i^t \mathbf x_i(0) + \sum_{j=0}^{t-1} \mathbf A_i^{t-j-1} \mathbf B_i \mathbf u_i(j) \Bigg) \Bigg.\Bigg] - \mathbf R_i \mathbf u_i(k) .
	\end{multline}
	Multiply \eqref{eq21} by $\mathbf u_i^\top(k)$ and then taking the summation over $i$, where $i \in \{1, \dots, n\}$, we obtain
	\begin{equation}\label{eq22}
		\resizebox {0.487\textwidth} {!} {$
			\begin{gathered}
				\lambda^\ast_k \sum_{i=1}^{n} \mathbf u_i^\top(k) \mathbf H_i \mathbf u_i(k) = -\sum_{t=k+1}^{N} \Bigg[ \sum_{i=1}^{n} \mathbf u_i^\top(k) \Bigg.  \Bigg(\mathbf A_i^{t-k-1} \mathbf B_i \Bigg)^\top   \mathbf Q_i \\ \times \Bigg( \mathbf A_i^t \mathbf x_i(0) + \sum_{j=0}^{t-1} \mathbf A_i^{t-j-1} \mathbf B_i \mathbf u_i(j) \Bigg)  \Bigg.\Bigg] - \sum_{i=1}^{n} \mathbf u_i^\top(k) \mathbf R_i \mathbf u_i(k).
			\end{gathered}$}
	\end{equation}
	Substitute $\sum_{i=1}^{n} \mathbf u_i^{\top}(k) \mathbf H_i \mathbf u_i(k) = C(k)$ in \eqref{eq3_2} into \eqref{eq22}, then
			\begin{equation}\label{eq25}
				\small
				\resizebox {0.487\textwidth} {!} {$
					\begin{aligned}
						&\lambda^\ast_k = -\frac{1}{C(k)}  \sum_{t=k+1}^{N} \Bigg[ \sum_{i=1}^{n}  \Bigg.  \Bigg(\mathbf A_i^{t-k-1} \mathbf B_i \mathbf u_i(k) \Bigg)^\top \mathbf Q_i \Bigg( \mathbf A_i^t \mathbf x_i(0) \\&+ \sum_{j=0}^{t-1} \mathbf A_i^{t-j-1} \mathbf B_i \mathbf u_i(j) \Bigg)  \Bigg.\Bigg]- \frac{1}{C(k)}\sum_{i=1}^{n} \mathbf u_i^\top(k) \mathbf R_i \mathbf u_i(k).
					\end{aligned}$}
			\end{equation}	
			Since $\mathbf Q_i>0$, we obtain $$-(\mathbf A_i^{t-k-1}\mathbf B_i \mathbf u_i(k))^\top \mathbf Q_i \mathbf A_i^{t-k-1} \mathbf B_i \mathbf u_i(k)\leq 0.$$ Similarly, $\mathbf R_i>0$, and  $- \mathbf u_i^\top(k) \mathbf R_i \mathbf u_i(k) \leq 0$. 
			
			Next, we seek an upper bound for $\lambda^\ast_k$. First, let these two terms equal zero. Then use Assumption \ref{assumption3} and the inequality in \eqref{eq3_6}, and by substitution into  \eqref{eq25},  we yield
			\begin{multline}\label{eq24}
				\lambda^\ast_k \leq \frac{1}{C(k)}  \sum_{t=k+1}^{N}  \Bigg[ \Bigg. n \alpha^{2t-k-1}\beta \sqrt{\frac{C(k)}{\rho}}  \delta_{\rm max} \gamma  \\ + n \beta^2 \sqrt{\frac{C(k)}{\rho}} \delta_{\rm max}\sum_{\substack{ j=0 \\ j\ne k}}^{t-1} \alpha^{2t-j-k-2} \sqrt{\frac{C(j)}{\rho}} \Bigg.\Bigg].
			\end{multline}
			By assumption, the right-hand side of \eqref{eq24} is less than or equal to $\lambda^\dag$. Therefore, we obtain $\lambda^\ast_k \leq \lambda^\dag$.
		\end{proof}

		\subsection{Dynamic Programming Approach}
		Similar to Section \ref{sec:QuadProg}, we only consider $\lambda^\ast_t >0$ and we deal with the optimization problem in \eqref{opt_DLTD_2} which is equivalent to
		\begin{equation}\label{opt_DLTD_2_3}
			\begin{aligned}
				\max_{{\mathbf U}_i} \quad   &- \mathbf x_i^\top(N) \mathbf Q_i \mathbf x_i(N) + \sum_{t=0}^{N-1} \big[ \big. - \mathbf x_i^\top(t) \mathbf Q_i \mathbf x_i(t) \\ &- \mathbf u_i^\top(t) \bigg(\mathbf R_i + \lambda^\ast_t \mathbf H_i \bigg) \mathbf u_i(t)+\lambda_t^\ast  a_i(t) \big. \big] \\
				{\rm s.t.} \quad & \mathbf x_i(t+1) = \mathbf A_i \mathbf x_i(t)+ \mathbf B_i \mathbf u_i(t), \quad t \in \mathcal{T}.
			\end{aligned}
		\end{equation}
		
		\medskip
		
		\begin{theorem}\label{theorem3}
			Consider the dynamic MAS on the time horizon ${N}$. Let Assumptions \ref{assumption2} and \ref{assumption3} hold. Suppose $\delta_{\rm max} \in \mathbb{R}^{>0}$ is selected from the following set
			\begin{multline*}\label{eq_set_Q2}
					\mathscr{S}_{\ast} =\left\{ \Bigg. \right. \delta_{\rm max} \in \mathbb{R}^{>0} : 
					\delta_{\rm max} \sum_{t=1}^{N}  \gamma \alpha^{2t-1}    \leq \frac{\sqrt{C(0)\rho}}{n \beta}  \lambda^\dag,\\
					\delta_{\rm max} \sum_{t=k+1}^{N}  \Bigg[ \Bigg. \gamma \alpha^{2t-k-1}   + \beta  {\sum_{\substack{ j=0 }}^{k-1}} \sqrt{\frac{C(j)}{\rho}} \alpha^{2t-j-k-2} \Bigg.\Bigg] \\  \leq \frac{\sqrt{C(k)\rho}}{n \beta}  \lambda^\dag \, \, \, \,  {\rm for}\,\, \forall k \in \mathcal{T}, k\neq 0\left. \Bigg. \right\}.
			\end{multline*} 
			Then the resulting  $\boldsymbol \lambda^\ast$ is socially resilient for all utility functions satisfying $\left\| \mathbf Q_i \right\| \leq \delta_{\rm max}$  (or $\mathbf Q_i \leq \delta_{\rm max} \mathbf I$).
		\end{theorem}
		
		\begin{proof}
			Similar to the previous section,  \eqref{eq3_2} and \eqref{eq3_6} are satisfied. 
			Since $\lambda^\ast_t > 0$, we obtain $ \mathbf R_i + \lambda^\ast_t \mathbf H_i >0$. Consequently, the optimization problem in \eqref{opt_DLTD_2_3} is a \emph{standard LQR problem}. Therefore, at each time step $k \in \mathcal{T}$, the optimal control solution is obtained as
			\begin{equation*}\label{eq29}
				\resizebox {0.487\textwidth} {!} {$
						\mathbf u_i^\ast(k) = - \big(\mathbf B_i^\top \mathbf P_{i, k+1} \mathbf B_i + ( \mathbf R_i + \lambda^\ast_k \mathbf H_i) \big)^{-1} \mathbf B_i^\top \mathbf P_{i, k+1} \mathbf A_i \mathbf x_i(k),
					$}
			\end{equation*}
			where 
			\begin{multline}\label{eq_p}
				\mathbf P_{i,k} =\mathbf A_i^\top \mathbf P_{i,k+1} \mathbf A_i + \mathbf Q_i - \mathbf A_i^\top \mathbf P_{i,k+1} \mathbf B_i \big( \mathbf B_i^\top \mathbf P_{i, k+1} \mathbf B_i \\ + (\mathbf R_i + \lambda^\ast_k \mathbf H_i) \big)^{-1} \mathbf B_i^\top \mathbf P_{i, k+1} \mathbf A_i.
			\end{multline}
			The Riccati difference equation in \eqref{eq_p} is initialized with $\mathbf P_{i,N}= \mathbf Q_i$ and is solved backward from $k=N-1$ to $k=0$. Additionally, since the last term on the right-hand side of \eqref{eq_p} is negative semi-definite (NSD), we obtain
			\begin{equation*}\label{eq35}
				\mathbf P_{i,k} \leq \mathbf A_i^\top \mathbf P_{i,k+1} \mathbf A_i + \mathbf Q_i,
			\end{equation*}
			and therefore,
			\begin{equation*}\label{eq36}
				\left\|\mathbf P_{i,k} \right\| \leq \alpha^2 \left\| \mathbf P_{i,k+1} \right\| + \left\| \mathbf Q_i \right\|.
			\end{equation*}
			Starting from $k=N-1$, we obtain
			\begin{equation}\label{eq37}
				\begin{aligned}
					\left\|\mathbf P_{i,N-1} \right\| \leq & (\alpha^2 +1 ) \left\| \mathbf Q_i \right\|, \\
					\vdots\\
					\left\|\mathbf P_{i,N-p} \right\| \leq & (\alpha^{2p} + \alpha^{2(p-1)}+ \dots+\alpha^2 +1 ) \left\| \mathbf Q_i \right\|.	
				\end{aligned}
			\end{equation}
			On the other hand, to find the optimal control input, we start from $k=0$ and proceed in a forward manner. In the first step, we get
			\begin{equation}\label{eq34}
					\mathbf u_i^\ast(0) = - \big(\mathbf B_i^\top \mathbf P_{i, 1} \mathbf B_i +  \mathbf R_i + \lambda^\ast_0 \mathbf H_i \big)^{-1} \mathbf B_i^\top \mathbf P_{i, 1} \mathbf A_i \mathbf x_i(0).
			\end{equation}
		Multiply \eqref{eq34} by $\mathbf u_i^{\ast \top}(0)\big(\mathbf B_i^\top \mathbf P_{i, 1} \mathbf B_i +  \mathbf R_i + \lambda^\ast_0 \mathbf H_i \big)$ and then taking the summation over $i$, where $i\in \{1, \dots, n\}$, we obtain
	\begin{multline}\label{eq38}
		\sum_{i=1}^{n}  \mathbf u_i^{\ast \top}(0)\big(\mathbf B_i^\top \mathbf P_{i, 1} \mathbf B_i +  \mathbf R_i + \lambda^\ast_0 \mathbf H_i \big) \mathbf u_i^\ast(0) \\ = - \sum_{i=1}^{n}\mathbf u_i^{\ast \top}(0) \mathbf B_i^\top \mathbf P_{i, 1} \mathbf A_i \mathbf x_i(0).
	\end{multline}
	Substitute \eqref{eq3_2} into \eqref{eq38}, and yield
				\begin{multline}\label{eq40}
					\small
					\lambda_{0}^\ast = - \frac{1}{C(0)} \sum_{i=1}^{n}\mathbf u_i^{\ast \top}(0) \mathbf B_i^\top \mathbf P_{i, 1} \mathbf A_i \mathbf x_i(0) \\- \frac{1}{C(0)}\sum_{i=1}^{n} \mathbf u_i^{\ast \top}(0)\big(\mathbf B_i^\top \mathbf P_{i, 1} \mathbf B_i +  \mathbf R_i \big) \mathbf u_i^\ast(0).
				\end{multline}
				To obtain an upper bound for $\lambda^\ast_0$, we can omit the second term on the right-hand side of \eqref{eq40} which is always non-positive. 
				Additionally,  using norm properties and considering Assumption \ref{assumption3} and the inequalities in \eqref{eq3_6} and \eqref{eq37}, we obtain 
				\begin{equation}\label{eq41}
					\begin{aligned}
						\lambda_{0}^\ast & \leq  \frac{n}{C(0)} \sqrt{\frac{C(0)}{\rho}} \beta  \gamma \delta_{\rm max} \sum_{t=1}^{N} \alpha^{2t-1}.
					\end{aligned}
				\end{equation}
				Moving forward to the time step $k>0$, we obtain
				\begin{multline}\label{eq48}
						\lambda^\ast_k \leq \frac{n}{\sqrt{C(k)\rho}} \beta \delta_{\rm max} \sum_{t=k+1}^{N}  \Bigg[ \Bigg. \gamma \alpha^{2t-k-1}   \\+ \beta   \sum_{\substack{ j=0 }}^{k-1} \sqrt{\frac{C(j)}{\rho}} \alpha^{2t-j-k-2} \Bigg.\Bigg]. 
				\end{multline}
				By assumption, the right-hand side of \eqref{eq41} and \eqref{eq48} is less than or equal to $\lambda^\dag$, which confirms $\lambda^\ast_k \leq \lambda^\dag$ for $k \in \mathcal{T}$.
			\end{proof}
			
			
			\section{Numerical Algorithm}\label{sec:NumericalAlgorithm}
			The two proposed sets in Theorems \ref{theorem2} and \ref{theorem3} are conservative but they give a concrete idea about how the trade-off between utility functions' parameters and the price threshold is achievable. To obtain more accurate and practical results, we propose a numerical algorithm  which provides less conservative bounds on the  parameters. This algorithm proceeds based on the bisection method.
			
			\medskip
			
			\noindent\textbf{Numerical Social Shaping Problem.} Consider the social welfare problem in \eqref{opt_social_DLD_1}. Let $\delta_{\rm max} \in \mathbb{R}^{>0}$ be the designing parameter. Suppose Assumption \ref{assumption2} holds with  $\mathbf Q_i = q_i \mathbf I$ where agents have the freedom to select $q_i \in \left(0, \delta_{\rm max} \right]$. Assume $\lambda^\dag$ is the given  price threshold accepted by all agents and $\mathbf R_i$ is specified for each $i \in \mathcal{V}$. We aim to find the upper bound $\delta_{\rm max}$ by a numerical approach such that if $q_i \in \left(0, \delta_{\rm max} \right]$ for $i \in \mathcal{V}$ then $\lambda^\ast_t \leq \lambda^\dag$ for $t \in \mathcal{T}$. The key steps to reach this purpose are illustrated in Algorithm \ref{algorithm1}.
			
			
			\begin{algorithm}[h]
				\SetAlgoLined
				\textbf{Input:} {System parameters $\mathbf A_i$, $\mathbf B_i$, and $\mathbf H_i$, the initial state $\mathbf x_i(0)$,  the time horizon $N$, the penalty matrix $\mathbf R_i$, and the local supply $a_i(t)$ for $i \in \mathcal{V}$ and $t \in \mathcal{T}$}.
				
				\textbf{Structure:} {Consider $\mathbf Q_i = q_i \mathbf I$}. Define 
				\begin{equation}\label{eq52}
					\bar \lambda^\ast(\delta) = \max_{q_1, \dots, q_n \in \left(0, \delta \right] }\max_{t \in \mathcal{T}} \lambda^\ast_t.
				\end{equation}
				
				\textbf{Initialize:} Set $k=0$, $b_0 = 0$, and $d_0 =d_\varrho>0$ such that $d_\varrho$ is sufficiently large to satisfy $\bar \lambda^\ast(d_\varrho)>\lambda^\dag$. 
				
				\While{\rm True}{
					$L_k=(b_k + d_k)/2$, \quad
					$\lambda_k=\bar \lambda^\ast(L_k)$;
					
					\uIf{$\lambda_k > \lambda^\dag$}{
						$b_{k+1} = b_k$ and $d_{k+1}=L_k$;\\
						$k=k+1$;
					}
					\uElseIf{$\lambda_k < \lambda^\dag$}{
						$b_{k+1} = L_k$  and $d_{k+1}=d_k$;\\
						$k=k+1$;
					}
					\Else{
						$\delta_{\rm max} = L_k$;\\
						{break}
					}
	
				}	
				\textbf{Output:} {$\delta_{\rm max}=L_k$} if the algorithm stops after a finite number of steps. Otherwise, $\delta_{\rm max}=\lim_{k \rightarrow \infty}L_k$.
				
				\caption{Bisection-Based  Social Shaping}
				\label{algorithm1}
			\end{algorithm}
			
			\begin{lemma}\label{lemma2}
				The function $\bar \lambda^\ast(\delta)$ in \eqref{eq52} is monotonically increasing.
			\end{lemma}
			\begin{proof}
				When $\delta$ increases, the domain of agents' preferences expands respectively.  Therefore, the maximum possible price can never decrease.
			\end{proof}
			
			\medskip
			
			\begin{theorem}
				The auxiliary variable $L_k$ in Algorithm \ref{algorithm1} converges to $L^\ast$ for some $L^\ast \in \left(0, d_\varrho \right)$ when $k \rightarrow \infty$. 
			\end{theorem}
			\begin{proof}
				From the update rules in Algorithm \ref{algorithm1}, we obtain 
				\begin{equation}\label{eq54}
					b_{k+1} \geq b_k, \quad \quad d_{k+1} \leq d_k,
				\end{equation} 
				and
				\begin{equation}\label{eq51}
					b_0\leq b_k < L_k < d_k \leq d_0.
				\end{equation}
				Inequalities in \eqref{eq54} and \eqref{eq51} imply that $b_k$ is monotonically increasing and  bounded above by $d_0=d_\varrho$. Similarly, $d_k$ is monotonically decreasing and  bounded below by $b_0=0$. Consequently, $b_k$ and $d_k$ converge when $k \rightarrow \infty$.
				
				Additionally, from the algorithmic steps, we obtain
				\begin{equation*}\label{eq49}
				| b_{k} - d_{k}| = 0.5^k d_\varrho.
				\end{equation*}
				Therefore, 
				\begin{equation}\label{eq50}
					\lim_{k \rightarrow \infty} |b_k - d_k| =0.
				\end{equation}
				Considering \eqref{eq54}, \eqref{eq51}, and \eqref{eq50}, we conclude 
				\begin{equation}\label{eq53}
					\lim_{k \rightarrow \infty} b_k=\lim_{k \rightarrow \infty} d_k= 	\lim_{k \rightarrow \infty} L_k =L^\ast,
				\end{equation}
				where $b_k < L^\ast< d_k$ at each iteration, and $L^\ast \in \left(0, d_\varrho \right)$.
				%
			\end{proof}
			
			\medskip

			\begin{theorem}\label{theorem4}
				Suppose there exists $\delta^\dag>0$ such that $\bar \lambda^\ast(\delta^\dag)=\lambda^\dag$. Then $\delta_{\rm max}$ obtained from Algorithm \ref{algorithm1} satisfies $\bar\lambda^\ast(\delta_{\rm max}) = \lambda^\dag$.
			\end{theorem}
			\begin{proof}
				If the algorithm stops  when $\lambda_k = \lambda^\dag$, we end with $\delta_{\rm max} = L_k$ and  $\bar \lambda^\ast (\delta_{\rm max}) = \lambda^\dag$. Otherwise, we obtain $\delta_{\rm max} =  \lim_{k \rightarrow \infty}L_k=L^\ast$.
				By contradiction, suppose $\bar \lambda^\ast(L^\ast) \neq \lambda^\dag$ leading to $L^\ast  \neq \delta^\dag$. First, consider $L^\ast < \delta^\dag$. According to  \eqref{eq53}, we obtain
				$$\exists  \, l>0 \quad {\rm s.t.} \quad d_k < \delta^\dag \quad {\rm for} \quad \forall \, k \geq l.$$
				Following Lemma \ref{lemma2}, we obtain $\bar \lambda^\ast(d_k) \leq \bar \lambda^\ast(\delta^\dag) = \lambda^\dag$, which contradicts the assumptions in Algorithm \ref{algorithm1}.
				On the other hand, if $L^\ast > \delta^\dag$, a similar analysis can be made for $b_k$ which leads to a contradiction. Consequently, it follows that $\bar \lambda^\ast(L^\ast) = \bar \lambda^\ast(\delta_{\rm max}) = \lambda^\dag$. 
			\end{proof}
			\section{Simulation Results}\label{sec:examples}
			\begin{example}
				Consider a dynamic MAS with $3$ agents who satisfy Assumptions \ref{assumption2} and \ref{assumption3} in the time horizon $N=6$, with the system parameters which are selected  as
				\begin{equation*}
					\small
					\begin{aligned}
						&{\mathbf A_1} = \left[ {\begin{array}{*{20}{c}}
								{ - 0.6}&{ - 0.1}&{0.2}\\
								{0.3}&{ - 0.7}&{0.2}\\
								{0.2}&{ - 0.3}&{0.8}
						\end{array}} \right],	\, {\mathbf B_1} = \left[ {\begin{array}{*{20}{c}}
								2&1\\
								1&7\\
								1&6
						\end{array}} \right], \\ &{\mathbf A_2} = \left[ {\begin{array}{*{20}{c}}
								{0.5}&{0.1}&{ - 0.1}\\
								{0.3}&{ - 0.2}&{ - 0.2}\\
								{ - 0.2}&{0.3}&{ - 0.3}
						\end{array}} \right], \, {\mathbf B_2} = \left[ {\begin{array}{*{20}{c}}
								4&1\\
								1&6\\
								3&4
						\end{array}} \right], \\ &{\mathbf A_3} = \left[ {\begin{array}{*{20}{c}}
								{ - 0.4}&{0.2}&{0.2}\\
								{ - 0.3}&{0.8}&{0.2}\\
								{0.2}&{0.3}&{ - 0.5}
						\end{array}} \right], \,
						{\mathbf B_3} = \left[ {\begin{array}{*{20}{c}}
								9&2\\
								2&3\\
								4&5
						\end{array}} \right],
					\end{aligned}
				\end{equation*}
				and the initial states 
				\begin{equation*}
					\small
					{\mathbf x_1}(0) = \left[ {\begin{array}{*{20}{c}}
							{10}\\
							{40}\\
							{70}
					\end{array}} \right], \, {\mathbf x_2}(0) = \left[ {\begin{array}{*{20}{c}}
							{20}\\
							{50}\\
							{80}
					\end{array}} \right], \, {\mathbf x_3}(0) = \left[ {\begin{array}{*{20}{c}}
							{30}\\
							{60}\\
							{90}
					\end{array}} \right].
				\end{equation*}
				Additionally, suppose agents have the representative local resources $a_1(t)=\sin(\frac{\pi}{6}t)+1.2$, $a_2(t)=2\sin(\frac{\pi}{6}t)+2.2$, and $a_3(t)=0$. The total network generation is obtained as $C(t)=\sum_{i=1}^{3}a_i(t)=3\sin(\frac{\pi}{6}t)+3.4$ and depicted in Fig. \ref{fig_1}.
				\begin{figure}[!t]
					\centering
					\centerline{\psfig{figure=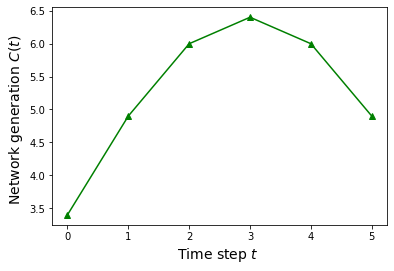,width=2.55in}}
					\caption{The network generation $C(t)$ over time steps.}
					\label{fig_1}
				\end{figure}
				Also, let $\lambda^\dag =30$, 
				$\mathbf R_1 = \mathbf R_2 = \mathbf R_3 = 0.05 \mathbf I$, and
				\begin{equation*}
					\small
					\begin{split}
						{\mathbf H_1} = \left[ {\begin{array}{*{20}{c}}
								5&1\\
								1&8
						\end{array}} \right], \, {\mathbf H_2} = \left[ {\begin{array}{*{20}{c}}
								3&2\\
								2&7
						\end{array}} \right], \, {\mathbf H_3} = \left[ {\begin{array}{*{20}{c}}
								2&{-1}\\
								{-1}&1
						\end{array}} \right].
					\end{split}
				\end{equation*}
				We aim to design $\mathbf Q_i$ for $i \in \mathcal{V}$ such that $\lambda^\ast_t \leq \lambda^\dag$ at all time steps. 
				
				\medskip
				\textit{Analytical approach:} For this example, the upper bound $\delta_{\rm max}$ of the personalized parameter $\mathbf Q_i$ is obtained from Theorems \ref{theorem2} and \ref{theorem3} as $0.00055$ and $0.0010$, respectively. We observe that Theorem \ref{theorem3} provides a larger upper bound compared to Theorem \ref{theorem2} so we assign $\delta_{\rm max}=0.001$. Next, we must select $\mathbf Q_i$ such that $\left\| \mathbf Q_i \right\| \leq \delta_{\rm max}$  (or $\mathbf Q_i \leq \delta_{\rm max} \mathbf I$) for $i \in \mathcal{V}$. Let us carefully choose $\mathbf Q_1 = \mathbf Q_2 = \mathbf Q_3 = 0.001 \mathbf I$. The optimal prices obtained from solving the social welfare problem in \eqref{opt_social_DLD_1} are depicted in Fig. \ref{fig_2} at different time steps.
			\begin{figure}[t]
				\centering
				\begin{subfigure}{.75\columnwidth}
					\centering
					\includegraphics[width=\linewidth]{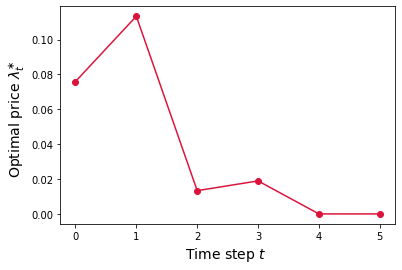}
					\caption{ $\delta_{\rm max}=0.001$, obtained from Theorem \ref{theorem3}.}
					\label{fig_2}
				\end{subfigure}%
			\quad
				\hfill
				\begin{subfigure}{.75\columnwidth}
					\centering
					\includegraphics[width=\linewidth]{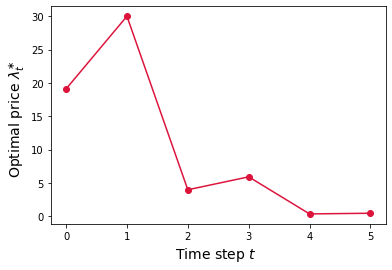}
					\caption{$\delta_{\rm max}=0.202$, obtained from Algorithm \ref{algorithm1}.}
					\label{fig_3}
				\end{subfigure}
				\caption{The optimal price $\lambda^\ast_t$ over time steps.}
			\end{figure}
				As indicated, the optimal prices are much less than $30$, and therefore, socially resilient; this confirms that Theorems \ref{theorem2} and \ref{theorem3} are valid but they provide conservative results.
				
				\medskip
				\textit{Numerical approach:} On the other hand, we run  Algorithm \ref{algorithm1} for 20 steps with the choice of $d_\varrho=1$. This value of $d_\varrho$ is sufficiently large that satisfies $\bar \lambda^\ast(d_\varrho)=148.6> \lambda^\dag$. The output of the algorithm is $\delta_{\rm max}=0.202$. Setting $\mathbf Q_1 = \mathbf Q_2 = \mathbf Q_3= 0.202 \mathbf I$, the obtained optimal prices  are as Fig. \ref{fig_3}, which are less than or equal to $30$ and socially resilient. 
				The maximum value of the price throughout the entire time horizon is $\lambda^\ast_1=30.0$ happening at time step $t=1$. If we select $\mathbf Q_i = q_i \mathbf I$ such that $q_i > 0.202$ for $i \in \mathcal{V}$, then we obtain $\lambda^\ast_1 > 30$, which is not socially acceptable. This shows the proposed numerical algorithm works well in practice. 
			\end{example}
			\section{CONCLUSIONS AND FUTURE WORKS}\label{sec:conclusions}
			In this paper, we have studied the issue of social shaping for dynamic MAS  over finite horizons. The system under consideration was a self-sustained dynamic MAS with distributed resource allocations operating at a competitive equilibrium. We presented a conceptual scheme which shows  how the social shaping problem is solvable implicitly under some convexity assumptions. As a typical case in the literature, we examined quadratic MAS. Dealing with an LQR problem using quadratic programming and dynamic programming,  we proposed two sets of quadratic utility functions under which the resource pricing at the competitive equilibrium is guaranteed to be socially acceptable, i.e., below a prescribed threshold.  Finally, we presented a numerical algorithm which provides more accurate bounds on the agents' preferences compared to the proposed analytical sets. As future work, it is suggested to extend the results to the infinite horizon case and consider network constraints in the framework.


		\end{document}